\documentclass{entcs} 
\usepackage{entcsmacro}
\usepackage{graphicx}
\usepackage{amsmath}
\usepackage[latin1]{inputenc}

\newcommand{\kems}{KEMS}

\usepackage{ifthen}
\newcommand{\troca}[3]
{\ifthenelse{\equal{#1}{1}}{#2}{#3}}

\newcommand{\EQUALDEF}{\stackrel{\mbox{\tiny def}}{=}}

\newcommand{\impli}{\to}

\newcommand{\NOT}{\neg}
\newcommand{\IMP}{\impli}
\newcommand{\AND}{\wedge}
\newcommand{\OR}{\vee}

\newcommand{\cons}{\circ}
\newcommand{\CONS}{\circ}

\newcommand{\com}[1]{}

\newcommand{\comm}[1]{{\bf }}

\newcommand{\T}{\begin{tt}T\end{tt}\, }
\newcommand{\F}{\begin{tt}F\end{tt}\, }
\newcommand{\x}{\ensuremath{\times}}

\newcommand{\mbc}{{\bf mbC}}
\newcommand{\cpl}{{\bf CPL}}

\newcommand{\mci}{{\bf mCi}}

\newcommand{\cum}{\ensuremath{C_1}}
\newcommand{\cene}{{\bf $C_n$}}

\newcommand{\ke}{{\bf KE}}

\newcommand{\bigor}{\bigvee}		
\newcommand{\bigand}{\bigwedge}

\newcommand{\FifF}{\mbox{\em $\Phi^5$}}		
\newcommand{\SF}{\mbox{\em $\Phi^6$}}		

\newcommand{\cumDef}{$\cons A \stackrel{\mbox{\tiny def}}{=} \neg (A \AND \neg A)$}

\newcommand{\satt}{$DS$}

\newcommand{\oPoCRule}[2]{
$
\begin{array}{c}
#1
\\ \hline
#2
\end{array}
$
}

\newcommand{\oPtCRule}[3]{
$
\begin{array}{c}
#1
\\ \hline
#2
\\ 
#3
\end{array}
$
}

\newcommand{\tPoCRule}[3]{
$
\begin{array}{c}
#1
\\ 
#2
\\ \hline  
#3
\end{array}
$
}

\newcommand{\thPoCRule}[4]{
$
\begin{array}{c}
#1
\\ 
#2
\\
#3
\\ \hline
#4
\end{array}
$
}

\newcommand{\connspcimp}{\!\!\!}

\newcommand{\timpOneNm}{($\T\connspcimp\IMP_1$)}
\newcommand{\timpOne}{
$
\tPoCRule{\T A \IMP B}{\T A}{\T B} \timpOneNm
$
}

\newcommand{\timpTwoNm}{($\T\connspcimp\IMP_2$)}
\newcommand{\timpTwo}{
$
\tPoCRule{\T A \IMP B}{\F B}{\F A} \timpTwoNm
$
}

\newcommand{\fimpNm}{($\F\connspcimp\IMP$)}
\newcommand{\fimp}{
$
\oPtCRule{\F A \IMP B}{\T A} {\F B} \fimpNm
$
}

\newcommand{\connspc}{\!}

\newcommand{\fandOneNm}{($\F\connspc\AND_1$)}
\newcommand{\fandOne}{
$
\tPoCRule{\F A \AND B}{\T A}{\F B} \fandOneNm
$
}

\newcommand{\fandTwoNm}{($\F\connspc\AND_2$)}
\newcommand{\fandTwo}{
$
\tPoCRule{\F A \AND B}{\T B}{\F A} \fandTwoNm
$
}

\newcommand{\tandNm}{($\T\connspc\AND$)}
\newcommand{\tand}{
$
\oPtCRule{\T A \AND B}{\T A} {\T B} \tandNm
$
}

\newcommand{\torOneNm}{($\T\connspc\OR_1$)}
\newcommand{\torOne}{
$
\tPoCRule{\T A \OR B}{\F A}{\T B} \torOneNm
$
}

\newcommand{\torTwoNm}{($\T\connspc\OR_2$)}
\newcommand{\torTwo}{
$
\tPoCRule{\T A \OR B}{\F B}{\T A} \torTwoNm
$
}

\newcommand{\forNm}{($\F\connspc\OR$)}
\newcommand{\for}{
$
\oPtCRule{\F A \OR B}{\F A} {\F B} \forNm
$
}

\newcommand{\fnegNm}{($\F\connspc\NOT$)}
\newcommand{\fneg}{
$
\oPoCRule{\F \NOT A}{\T A} \fnegNm
$
}

\newcommand{\pbNm}{(PB)}

\newcommand{\newPB}
{
\Tree [.{} $\T\,A$ $\F\,A$ ] (PB)
}

\newcommand{\tnegMbcNm}{($\T\connspc\NOT'$)}

\newcommand{\tnegCumNm}{($\T\connspc\cons\neg$)}
\newcommand{\tnegCum}{
$
\tPoCRule{\T \cons A}{\T \NOT A}{\F A} \tnegCumNm
$
}

\newcommand{\tnotnotNm}{($\T\connspc\NOT\NOT$)}
\newcommand{\tnotnot}{
$
\oPoCRule{\T \NOT \NOT A}{\T A} \tnotnotNm
$
}

\newcommand{\any}{ \oslash }

\newcommand{\fconsanyOneNm}{($\F\connspc\cons\any_1$)}
\newcommand{\fconsanyOne}{
$
\tPoCRule{\F \cons (A \any B)}{\T \cons A}{\F \cons B} \fconsanyOneNm
$
}

\newcommand{\fconsanyTwoNm}{($\F\connspc\cons\any_2$)}
\newcommand{\fconsanyTwo}{
$
\tPoCRule{\F \cons (A \any B)}{\T \cons B}{\F \cons A} \fconsanyTwoNm
$
}

\newcommand{\tNotAnyLeftNm}{($\T\NOT\any_1$)}
\newcommand{\tNotAnyLeft}{
$
\thPoCRule{\T \NOT (A \any B)}{\T A \any B}{\T \cons A}{\F \cons B} \tNotAnyLeftNm
$
}

\newcommand{\tNotAnyRightNm}{($\T\NOT\any_2$)}
\newcommand{\tNotAnyRight}{
$
\thPoCRule{\T \NOT (A \any B)}{\T A \any B}{\T \cons B}{\F \cons A} \tNotAnyRightNm
$
}

\newcommand{\tNotAnyLeftPBNm}{($\T\NOT\any_1'$)}
\newcommand{\tNotAnyLeftPB}{
$
\tPoCRule{\T \NOT (A \any B)}{\T (A \any B) \AND \cons A}{\F \cons B} \tNotAnyLeftPBNm
$
}

\newcommand{\tNotAnyRightPBNm}{($\T\NOT\any_2'$)}
\newcommand{\tNotAnyRightPB}{
$
\tPoCRule{\T \NOT (A \any B)}{\T (A \any B) \AND \cons B}{\F \cons A} \tNotAnyRightPBNm
$
}

\usepackage{url}

\usepackage{qtree}

\def\lastname{Neto}
\begin{document}

\begin{frontmatter}
%A theorem prover for a paraconsistent logic - A KE System for \cum
\title{Towards an efficient prover for the \cum\/ paraconsistent logic}

\author{Adolfo Neto\thanksref{Consrel}\thanksref{myemail}}
\address{Informatics Department (DAINF)\\Federal University of Technology - Paran\'a (UTFPR)\\ Curitiba, Brazil}
\author{Celso A.~A.~Kaestner\thanksref{coemail}}
\address{Informatics Department (DAINF)\\Federal University of Technology - Paran\'a  (UTFPR)\\ Curitiba, Brazil}
\author{Marcelo Finger\thanksref{Consrel}\thanksref{cocoemail}}
\address{Computer Science Department (DCC)\\University of S\~{a}o Paulo (USP)\\S\~{a}o Paulo, Brazil}
\thanks[Consrel]{Adolfo Neto and Marcelo Finger acknowledge support from Funda\c{c}c\~{a}o de Amparo a Pesquisa do Estado de S\~{a}o Paulo (FAPESP), Brazil, through the Thematic Project ConsRel, 
grant number 2004/14107-2. Marcelo Finger also acknowledges an individual research grant from the Brazilian Research Council (CNPq).}
\thanks[myemail]{Email:
    \href{mailto:adolfo@utfpr.edu.br} {\texttt{\normalshape
        adolfo@utfpr.edu.br}}} 
\thanks[coemail]{Email:
    \href{mailto:celsokaestner@utfpr.edu.br} {\texttt{\normalshape
        celsokaestner@utfpr.edu.br}}}
\thanks[cocoemail]{Email:
    \href{mailto:mfinger@ime.usp.br} {\texttt{\normalshape
        mfinger@ime.usp.br}}}

\begin{abstract}
%      State the problem
The \ke\/ inference system is a tableau method developed by Marco Mondadori
which was presented as an improvement, in the computational efficiency sense, over 
Analytic Tableaux.
In the literature, there is no description of a theorem prover based on the \ke\/ method 
%\cite{DAgostino99} 
for the \cum\/ paraconsistent logic.
%\cite{teseDaCosta} 
%      Say why it's an interesting problem
%A \ke-based theorem prover for \cum\/ would probably be more efficient than existing provers for \cum. 
Paraconsistent logics have several applications, such as in robot control and medicine.
These applications could benefit from the existence of such a prover.
%      Say what your solution achieves
We present a sound and complete \ke\/ system for \cum, an informal specification of a 
%\cite{KEMSsite} 
strategy for the \cum\/ prover
as well as problem families that can be used to evaluate provers for \cum.
%      Say what follows from your solution
The \cum~\ke\/ system and the strategy described in this paper 
will be used to implement a \ke~based prover for \cum, which
will be useful for those who study and apply paraconsistent logics.
%This prover will be tested with the problem families presented in this paper.
\end{abstract}


%\title{\cum\/ \ke\/}

 \end{frontmatter}

\section{Introduction}

Inconsistency is a phenomena that appears naturally in our world.
Consider the following situation: two persons have different
(contradictory) opinions about a specific statement $A$: the first
one considers $A$ true, meanwhile the second one believes that $\neg
A$ is true. This contradiction, however, should not prevent that
common conclusions which do not involve $A$ -- directly or indirectly
-- can be deduced.

This situation is not adequately managed by classical logic, since
it is not equipped to deal with inconsistency. The reason is the
well known {\em ``Ex contradictione sequitur quod libet''} principle: if
a theory $\Gamma$ is inconsistent, that is, if formulas $A$ and
$\neg A$ are theorems, then every formula $B$ of the language is
also a theorem in $\Gamma$; or, shortly, $\Gamma$ becomes
{\em trivial}.

Paraconsistent Logics were initially proposed by Da Costa \cite{teseDaCosta} as
logical systems that deal with contradictions in a discriminating
way, avoiding the previous principle and managing inconsistent but
non-trivial theories.

Presently automatic proof methods are widely used in several
computer applications, such as in robot control \cite{4153157}, in medicine
\cite{Carvalho,10.1109/CBMS.2007.71}, and many others \cite{CKB}. Most of the employed methods work on
logical formalisms based on classical logic. In this paper we
present the specification of an strategy for automatic theorem prover based on a \ke\/ system, an
improvement of the well known tableaux deduction method, for a
particular paraconsistent logic called \cum.

The rest of this paper is organized as follows: Section 2 introduces
the axiomatization and valuation of the paraconsistent logic 
\cum; in Sections 3 and 4 we present the \ke\/ system for
\cum\/ and its inference rules, and the \kems\/ strategy,
respectively; Section 5 presents a set of problems constructed to
evaluate the prover; in Section 6 we present a motivating example,
showing that our proposal is adequate to deal with practical
problems; in Section 7 we compare our work with similar ones; finally
in Section 8 we draw some conclusions and propose future research.

We emphasize the main contributions of this paper: (a) a sound and
complete \ke\/ system for \cum\/ (Section~\ref{ke_cum}); (b) an informal
specification of a \kems\/ \cite{KEMSsite} strategy for the \cum\/ prover
(Section~\ref{strategy_cum}); and (c) problem families that can be used to evaluate
provers for \cum\/ (Section~\ref{problemFamilies}).

\subsection{Preliminaries}

Let ${\cal P}$ be a countable set of propositional letters. We concentrate on the propositional language
${\cal L}$ formed by the usual boolean connectives $\IMP$ (implication), $\AND$ (conjunction), $\OR$ (disjunction) and $\NOT$ (negation). We call $\Sigma$ this set of connectives: $\Sigma=\{ \NOT, \AND, \OR, \IMP\}$ ($\Sigma$ is called a signature in \cite{CCM05}). $\bigand^n_{i=1}$ and $\bigor^n_{i=1}$ are, respectively, iterated conjuntion and iterated disjunction.

Throughout the paper, we use uppercase Latin or lowercase Greek letters to denote arbitrary formulas, 
and uppercase Greek letters to denote sets of formulas. 

We also work here with signed formulas. A signed formula is an expression
${\cal S}~A$ where ${\cal S}$ is called the sign and $A$ is a propositional formula.
The symbols $\T$ and $\F$, respectively representing the `true' and `false' truth-values, can be used as signs.
The conjugate of a signed formula $\T A$ ($\F A$) is $\F A$ ($\T A$). The subformulas of a signed formula
${\cal S}~A$ are all the formulas of the form $\T B$ or $\F B$ where $B$ is a subformula of $A$.

The size of a signed formula ${\cal S}~A$ is defined as the size of $A$. 
The size $s(A)$ of a formula $A$ is defined as usual:
\begin{itemize}
 \item $s(A)=1$ if $A$ is a propositional atom;
 \item $s(\any A)=1+s(A)$, where $A$ is a formula and $\any$ is a unary connective;
 \item $s(A \any B)= 1 + s(A) + s(B)$, where $\any$ is a binary connective, and $A$ and $B$ are
      formulas.
\end{itemize}

A propositional valuation $v$ is a function $v: {\cal P} \rightarrow \{0,1\}$. We extend the definition of
valuations to signed formulas in the following way: $v(\T A)=v(A)$ and $v(\F A)=1-v(A)$.

\section{\cum, a paraconsistent logic}

% \begin{quotation}
% The hierarchy of logics \cene, $1 \leq n < \omega$, is a family of paraconsistent logics based
% on classical logic \cite{costa-paraconsistent}. These logics are also \lfi s \cite{CCM05}.
% Here we are only interested in the simplest logic of this family: \cum.
% In \cum, the consistency connective $\circ$ is an abbreviation, not a primitive connective:
% \[
% %%%%\circ A \stackrel{\mbox{\tiny def}}{=} \neg (A \AND \neg A).
% \circ A \EQUALDEF \neg (A \AND \neg A).
% \]
% \end{quotation}

\cum\/ is a paraconsistent logic \cite{teseDaCosta}, ``a logic of the early paraconsistent vintage'' \cite{CCM05}.
It is part of the hierarchy of logics \cene, $1 \leq n < \omega$  \cite{CKB}.
 \cum\/ is of historical importance because it was one of the first paraconsistent logics to be presented. 

% For a historical perspective on the birth of paraconsistent logic, the reader is referred to  \cite{teseMarcos}.

Paraconsistent logics are logics in which theories can be inconsistent but nontrivial \cite{CKB}.
In classical logic, $A \AND \NOT A \vdash B$ for any formulas $A$ and $B$. This is not true in paraconsistent logics. 
\com{
Paraconsistent logics have several applications, such as in robot control \cite{4153157}, in medicine \cite{10.1109/CBMS.2007.71,Carvalho}, and many others \cite{10.1109/CBMS.2007.71,Carvalho}. 
}

\com{
{\em
In classical logic, if you have a theory $\Gamma$, and with this theory you can prove a formula $A$
($\Gamma \vdash A$), you have to make sure that you do not have $\Gamma \vdash \NOT A$ (that is, $ \Gamma \not\vdash \NOT A$). Because if $\Gamma \vdash \NOT A$, then $\Gamma \vdash B$, for any formula $B$.
This feature is called the ``Principle of Explosion'' \cite{CCM05}.

In a paraconsistent logic such as \cum, it is possible to have $\Gamma \vdash A$ and $\Gamma \vdash \NOT A$ without having  $\Gamma \vdash B$ for any $B$. Therefore, the Principle of Explosion does not hold in \cum.
}
}

In \cum, a consistency operator ($\cons$) is introduced. The intended meaning of $\cons A$ is ``$A$ is consistent'' \cite{CCM05}.
According to \cite{CCM05}, ``da Costa's intuition was that the `consistency' (which he dubbed `good behavior')
of a given formula would not only be a sufficient requisite to guarantee
its explosive character, but that it could also be represented as an {\em ordinary
formula} of the underlying language.''

In \cum, da Costa represented the consistency of a formula $A$ by the formula $\NOT (A \AND \NOT A)$.
That is, the consistency connective ``$\circ$'' is not a primitive connective,  but an abbreviation:
\[
%%%%\circ A \stackrel{\mbox{\tiny def}}{=} \neg (A \AND \neg A).
\circ A \EQUALDEF \neg (A \AND \neg A).
\]

% In \cum, da Costa represented the consistency of a formula $A$ by the formula $\NOT (A \AND \NOT A)$.
% That is, $\cons A \EQUALDEF \NOT (A \AND \NOT A)$.
% %, \cons is used to abbreviate formulas.
% %da Costa also referred to this last formula as a realization of the `Principle of Non-Contradiction'.

% From \cite{CCM05}:
% \begin{quotation}
%  When proposing his first paraconsistent logics (cf. \cite{teseDaCosta}) da
% Costa's intuition was that the `consistency' (which he dubbed `good behavior')
% of a given formula would not only be a sufficient requisite to guarantee
% its explosive character, but that it could also be represented as an ordinary
% formula of the underlying language. For his initial logic, \cum, he chose to represent
% the consistency of a formula $\alpha$ by the formula $\NOT (\alpha \AND \NOT \alpha)$, and referred
% to this last formula as a realization of the `Principle of Non-Contradiction'.
% \end{quotation}

\com{
In \cum, $\Gamma \vdash A \AND \NOT A \AND \cons A$ is not possible 
unless $\Gamma \vdash B$ for any $B$. In such a case, $\Gamma$ is trivial.

In the next two subsections we will characterize \cum\/ by presenting its axiomatization and semantics.
}

% If you have 
% \[
%  \Gamma \vdash A \AND \cons A
% \]
% can you be sure that 
% \[
%  \Gamma \not\vdash \NOT A
% \]
% ??

\subsection{\cum's Axiomatization}

Some axiomatizations for \cum\/ were presented in the literature \cite{CCM05,daco:asem77,D'Ottaviano200627}.
The presentation below is based on \cite{CCM05} and \cite{daco:asem77}.

{\bf Axiom schemas:}
\begin{description}
 \item [(Ax1)] $\alpha \IMP (\beta \IMP \alpha)$
 \item [(Ax2)] $(\alpha \IMP \beta) \IMP ((\alpha \IMP (\beta \IMP \gamma)) \IMP (\alpha \IMP \gamma))$
 \item [(Ax3)] $\alpha \IMP (\beta \IMP (\alpha \AND \beta))$
 \item [(Ax4)] $(\alpha \AND \beta) \IMP \alpha$
 \item [(Ax5)] $(\alpha \AND \beta) \IMP \beta$
 \item [(Ax6)] $\alpha \IMP (\alpha \OR \beta)$
 \item [(Ax7)] $\beta \IMP (\alpha \OR \beta)$
 \item [(Ax8)] $(\alpha \IMP \gamma) \IMP ((\beta \IMP \gamma) \IMP ((\alpha \OR \beta) \IMP \gamma))$
% \item [(Ax9)] $\alpha \OR (\alpha \IMP \beta)$
 \item [(Ax10)] $\alpha \OR \NOT \alpha$
 \item [(Ax11)] $\NOT \NOT \alpha \IMP \alpha$
 \item [(bc1)] $\cons \alpha \IMP (\alpha \IMP (\NOT \alpha \IMP \beta))$
 \item [(ca1)] $(\cons \alpha \AND \cons \beta) \IMP \cons (\alpha \AND \beta)$
 \item [(ca2)] $(\cons \alpha \AND \cons \beta) \IMP \cons (\alpha \OR \beta)$
 \item [(ca3)] $(\cons \alpha \AND \cons \beta) \IMP \cons (\alpha \IMP \beta)$
\end{description}

{\bf Inference rule:}\\ \\
{\bf (MP)} $\dfrac{\alpha, \alpha \IMP \beta}{\beta}$

The difference from classical propositional logic (\cpl) axiomatization is that 
to obtain an axiomatization for \cpl\/ we must remove the schemas that deal with the consistency connective ({\bf (bc1), (ca1), (ca2)} and {\bf (ca3)}) and add the following axiom schema (called `explosion law'
in \cite{CCM05}):
\begin{description}
 \item [(exp)] $\alpha \IMP (\NOT \alpha \IMP \beta)$
\end{description}
% which 
%!!! It is interesting to compare \cum\/ axiomatization and semantics with CPL's.

% From \cite{daco:asem77,D'Ottaviano200627,CCM05}:
%From \cite{D'Ottaviano200627,CCM05}:
% 
% In \cite{CCM05} another axiom schema (which is redundant for \cum) was included:
% \begin{description}
%  \item [(Ax9)] $\alpha \OR (\alpha \IMP \beta)$
% \end{description}

\subsection{\cum's Valuation}

\cum\/ received a bivaluation semantics in \cite{daco:asem77} (see also \cite{CCM05}).
A set of clauses characterizing \cum-valuations (adapted from the one in \cite{CCM05}) 
\com{(see appendix~ \ref{discussion_cum_valuations})} 
is the following:
\begin{itemize}
 \item  $v(\alpha_1 \AND \alpha_2)=1 $ if and only if  $ v(\alpha_1)=1$ and $v(\alpha_2)=1$;
 \item  $v(\alpha_1 \OR \alpha_2)=1 $ if and only if $ v(\alpha_1)=1$ or $v(\alpha_2)=1$;
 \item  $v(\alpha_1 \IMP \alpha_2)=1 $ if and only if $ v(\alpha_1)=0$ or $v(\alpha_2)=1$;
 \item  $v(\NOT \alpha)=0 $ implies $ v(\alpha)=1$;
 \item  $v(\NOT \NOT \alpha)=1 $ implies $ v(\alpha)=1$;
 \item  $v(\cons \alpha)=1$  implies $v(\alpha)=0$ or $v(\NOT \alpha)=0$.
 \item  $v(\cons (\alpha \any \beta))=0 $ implies $ v(\cons \alpha)=0$ or $ v(\cons \beta)=0$, for $\any \in \{\AND, \OR, \IMP\}$;
% \item [(vc1)] $v(\alpha_1 \AND \alpha_2)=1 $ if and only if (iff) $ v(\alpha_1)=1$ and $v(\alpha_2)=1$;
% \item [(vc2)] $v(\alpha_1 \OR \alpha_2)=1 $ if and only if $ v(\alpha_1)=1$ or $v(\alpha_2)=1$;
% \item [(vc3)] $v(\alpha_1 \IMP \alpha_2)=1 $ if and only if $ v(\alpha_1)=0$ or $v(\alpha_2)=1$;
%  \item [(vc4)] $v(\NOT \alpha)=0 $ implies $ v(\alpha)=1$;
%  \item [(vc5)] $v(\NOT \NOT \alpha)=1 $ implies $ v(\alpha)=1$;
%  \item [(v5)] $v(\cons \alpha)=1$  implies $v(\alpha)=0$ or $v(\NOT \alpha)=0$.
%  \item [(vc7)] $v(\cons (\alpha \any \beta))=0 $ implies $ v(\cons \alpha)=0$ or $ v(\cons \beta)=0$, for $\any \in \{\AND, \OR, \IMP\}$;
\end{itemize}
%where $\cons \alpha$ abbreviates the formula $\NOT(\alpha \AND \NOT \alpha)$.

\begin{definition}
Let $\Gamma$ be $\{ A_1, A_2, \ldots, A_n \}$ for $n \ge 0$. 
$\Gamma \vdash B$ is a valid sequent in \cum~if and only if, 
$v(B)=1$ whenever $v(A_i)=1$ for all $i$ ($1 \leq i \leq n$).
``$\Gamma \vdash B$ is a valid sequent in \cum'' can be abbreviated to $\Gamma\vdash_{C_1}B$.
% ORIGINAL (BEFORE REVIEWER COMMENTS)
% Let $\Gamma$ be $\{ A_1, A_2, \ldots, A_n \}$ for $n \ge 0$. 
% $\Gamma \vdash B$ is a valid sequent in \cum~if and only if, 
% for all $i$ ($1 \le i \le n$), 
% when $v(A_i)=1$ then $v(B)=1$.
% ``$\Gamma \vdash B$ is a valid sequent in \cum'' can be abbreviated to $\Gamma\vdash_{C_1}B$.
\end{definition}

\section{The \ke\/ System for \cum}
\label{ke_cum}

The \ke\/ inference system is a tableau method \cite{dagostino94taming} developed by Marco Mondadori and discussed in detail in several works authored or co-authored by Marcello D'Agostino \cite{broda95solution,dagostino-are,DAgostino99}.
The \ke\/ system  was presented as an improvement,
in the computational efficiency sense, over Analytic Tableaux \cite{Smullyan68}.
A \ke\/ System is a tableau system in which there is only one branching rule.
As branching can lead to repetition of efforts (i.e.~the same work being done in two or more branches),  
branching rules lead to less efficient proof systems (and implementations) \cite{DAgostino99}.
\com{
It is a refutation system that, though close to the analytic tableau method, is not affected by the anomalies of cut-free systems \cite{DAgostino99}.
}
% Cut and pay? DBLP:journals/jolli/FingerG06

%\include{C1_KE_rules_estilo_tese}
%%% C1 KE RULES
\begin{figure}[t!]
% ORIGINAL \begin{figure}[htb]
\begin{center}
%\begin{small}
%\framebox{
$
\begin{array}{ccc}
\fimp & \timpOne & \timpTwo  \\
& & \\
\tand & \fandOne & \fandTwo  \\
& & \\
\for  & \torOne & \torTwo \\
& & \\
\fneg & \tnotnot & \\
& & \\
\tnegCum & \fconsanyOne & \fconsanyTwo   \\
& & \\
& \newPB &  \\
\end{array}
$
%}
%\end{small}
\end{center}
\caption{\cum\/ \ke\/ rules.}
\label{cum_KE_rules}
\end{figure}

\com{

%%% C1 KE RULES
\begin{figure}[t!]
% ORIGINAL \begin{figure}[htb]
\begin{center}
%\begin{small}
%\framebox{
$
\begin{array}{ccc}
\timpOne & \timpTwo & \fimp \\
& & \\
\fandOne & \fandTwo & \tand \\
& & \\
\torOne & \torTwo & \for \\
& & \\
\fneg & \newPB & \tnegCum \\
& & \\
\tnotnot & \tNotAnyLeft & \tNotAnyRight   \\
& & \\
     & \tNotAnyLeftPB & \tNotAnyRightPB   \\
\end{array}
$
%}
%\end{small}
\end{center}
\caption{New \cum\/ \ke\/ rules.}
\label{new_cum_KE_rules}
\end{figure}

{\bf 
In Figure~\ref{new_cum_KE_rules}, for the four $\T\NOT\any$ rules, $B$ is not $\NOT A$.

Therefore, every s-formula can be used as main at most once.
}

}

%\Tree [.{} $\T\,A$ $\F\,A$ ] (PB)
 
We present here a sound and complete
\ke\/ System we have devised for \cum.
The rules in our system are presented in Figure~\ref{cum_KE_rules}.
%The rules for the \cum\/ \ke\/ system are presented in Figure~\ref{cum_KE_rules}. 
Note that in \cum, the set of connectives is $\Sigma=\{ \NOT, \AND, \OR, \IMP\}$  but,
to make the rules simpler, we have used the connectives in $\Sigma^\CONS=\Sigma \cup \{ \CONS\}$, i.e.~including the consistency connective, which is actually an abbreviation.

In \cite{IFIPAI2006,NF07} (and also in \cite{teseAdolfo}) the first and third authors 
of this paper have presented \ke\/ Systems
for two other paraconsistent logics: \mbc\/ and \mci\/ (more about these two logics can be found in \cite{CCM05}). The \ke\/ System for \cum\/ has several rules in common with the \ke\/ Systems for
these two logics.  However, in these logics, consistency ($\cons$) is not a defined connective.

As in classical \ke\/  rules \cite{DAgostino99}, rules with ``1'' (for instance, ``$\F\AND_1$'') or  ``2'' as subscript are interchangeable. Only one of each pair is actually essential. By using the {\bf PB}
rule (Figure~\ref{cum_KE_rules}), the other can be derived.

% \begin{figure}
% \centerline{\pdfximage width 4in height 4in
%   {C1_KE_prova_exemplo_C3M.pdf}\pdfrefximage\pdflastximage}
% \caption{A proof of using \cum~tableau system\cite{CCM05}.}
% \label{prova_lfi}
% \end{figure}

Note also that 
%``\fconsanyOneNm'' 
$\F\!\!\CONS \any_1$
and 
%``\fconsanyTwoNm'' 
$\F\!\!\CONS \any_2$
are actually three rules each, because $\any$ can be any connective in $\{ \AND, \OR, \IMP \}$.
In a 2-premiss rule, the {\em main premiss} is the first premiss. The second premiss is
called {\em minor premiss}. The main premiss in 
%``\fconsanyOneNm''
$\F\!\!\CONS \any_1$
 (or in 
%``\fconsanyTwoNm'') 
$\F\!\!\CONS \any_2$)
can be ``$\F \cons (A \AND B)$'', ``$\F \cons (A \OR B)$'' or ``$\F \cons (A \IMP B)$''. 
And, as `$\cons$' is a defined connective, ``$\F \cons (A \any B)$'' is actually ``$\F \NOT ((A \any B)  \AND \NOT (A \any B)$''. For instance, $\F\!\!\CONS \AND_1$ is:
\begin{center}
\tPoCRule{\F \NOT ((A \AND B)  \AND \NOT (A \AND B))}{\T \NOT (A \AND \NOT A)}{\F \NOT (B \AND \NOT B)} $(\F\!\!\CONS \AND_1)$
\end{center}

It is easy to see that these rules ($\F\!\!\CONS \any_1$
and $\F\!\!\CONS \any_2$)
are not analytic. In $\F\!\!\CONS \AND_1$, $\F\!\!\NOT (B \AND \NOT B)$ is not a subformula of any premiss.

Therefore, in our system we have:
\begin{itemize}
 \item 12 essential linear  rules (5 of these rules are 1-premiss rules and 7 rules are 2-premiss rules);
 \item 6 derived linear 2-premiss rules;
 \item 1 (0-premiss) branching rule.
\end{itemize}
Of these rules, 6 of them ($\F\!\!\CONS \any_1$ and $\F\!\!\CONS \any_2$) are rather complex, 
far more complex than any \cpl\/ \ke\/ rule.

\begin{example}
The formula $\NOT(P \AND(\NOT P \AND \CONS P))$ 
can be proved in \cum~\ke~system as depicted in Figure~\ref{fig:prova_gamma_cum_ke}.
The same formula was proved in \cite{CCM05} using the \cum~tableau system presented there (Figure~\ref{fig:prova_lfi}). It is easy to see that the \cum~\ke~proof has less formula nodes and less branches
than the \cum~tableau \cite{CCM05} proof.

\begin{figure}
\begin{small}
\begin{center}
\[
\begin{array}{c}
\F\, \NOT (P \AND (\NOT P \AND \CONS P)) \\
\T\, P \AND(\NOT P \AND \CONS P)\\
\T\, P\\
\T\, \NOT P \AND \CONS P\\
\T\, \NOT P\\
\T\, \CONS P\\
\F\, P \\
\x
\end{array}
\]
\caption{A proof of $\NOT(P \AND(\NOT P \AND \CONS P))$ using the \cum~\ke~system.}
\label{fig:prova_gamma_cum_ke}
\end{center}
\end{small}
\end{figure}

\begin{figure}
\begin{small}
\begin{center}
\Tree
[.{$\F\, \NOT (P \AND (\NOT P \AND \CONS P))$ \\
$\T\, P \AND (\NOT P \AND \CONS P)$ \\
$\T\, P $ \\
$\T\, \NOT P \AND \CONS P$ \\
$\T\, \NOT P$ \\
$\T\, \CONS P$}
[.{$\T\, P \IMP P$}
[.{$\T\, P \IMP \NOT P$\\
   $\F\,P$\\
   $\x$
   }
]
[.{$\F\, P \IMP \NOT P$\\
   $\T\,P$\\
   $\F\,\NOT P$\\
   $\times$
   }
]
]
[.{$\F\, P \IMP P$\\
   $\T\,P$\\
   $\F\,P$\\
   $\times$
   }
]
]
\caption{A proof of $\NOT(P \AND(\NOT P \AND \CONS P))$ \cite{CCM05}.}
\label{fig:prova_lfi}
\end{center}
\end{small}
\end{figure}

\end{example}

\com{

We have proved 
\framebox{(Section~\ref{ke_cum_proof_sc})} 
that the \cum\/ \ke\/ system is correct and complete.
}

\com{

DISCUSS Signed formula (\T\/ and \F).

Rules with ``1'' (for instance, ``\fandOneNm'') and  ``2'' as subscripts are interchangeable. Only one of each pair is actually essential. Using PB, the other can be derived.

$A$ and $B$ can be any well formed formula.

$\any$ can be any connective in $\{ \AND, \OR, \IMP \}$.

$\cons$ is a derived connective in \cum. 

$\cons A \EQUALDEF \NOT (A \AND \NOT A)$.

How did we arrive at these rules?

First, we knew the \ke\/ System for CPL.

Then we got to know a tableau system for \mbc\/ and a tableau system for \cum.

After that we devised a \ke\/ System for \mbc.

Then we extended this system to get a \ke\/ System for \cum.

Note that ``\fconsanyOneNm'' and ``\fconsanyTwoNm'' are actually three rules each.

The main premiss in ``\fconsanyOneNm'' (or in ``\fconsanyTwoNm'') can be ``$\F \cons (A \AND B)$'', ``$\F \cons (A \OR B)$'' or ``$\F \cons (A \IMP B)$''.

And, as `$\cons$' is a defined connective, ``$\F \cons (A \any B)$'' is actually ``$\F \NOT ((A \any B)  \AND \NOT (A \any B)$''.

Therefore these rules are not analytic. See, for instance:

\tPoCRule{\F \NOT ((A \any B)  \AND \NOT (A \any B))}{\T \NOT (A \AND \NOT A)}{\F \NOT (B \AND \NOT B)} \fconsanyOneNm

And they are also very complex.

The size of the formulas in the ``\fconsanyOneNm'' rule is 19 (9+5+5). If ``$\cons$'' were not a defined connectivem it would be 8 (4+2+2).

For all the non-``$\NOT$'' and non-``$\cons$'' 1-premiss and 2-premiss rules, the complexity is 5 (3+1+1).

For ``\tnegMbcNm``,  the complexity is 8 (it would be 5 if ``$\cons$'' was not a defined connective).

``\fnegNm'' and ``\tnotnotNm'' have complexities of 3 and 4, respectively.

But, what does this mean for our implementation?

Does this make the premiss rule pattern matching more difficult?

Pattern matching is more difficult when you have to verify if a ``$\cons$'' rule can be applied.
Even more difficult when you think that:
\begin{itemize}
 \item if you translate all formulas with ``$\cons$'' to their defined meanings, the size of the problem may grow a lot (exponentially?)
 \item if you DO NOT translate, then whenever a $\NOT (A \AND \NOT A)$ appears, for any formula $A$ (no matter how complex it is), you should treat it as $\cons A$ in rule applications! 
\end{itemize}

}

\subsection{Soundness and Completeness}

Our intention here is to prove that the \cum\/ \ke\/ system is sound and complete. 
\com{
It is clear to us that the \cum\/ \ke\/ system is not analytic,
because its 
%\fconsanyOneNm\/ and \fconsanyOneNm\/ 
$\F\!\!\CONS \any_1$
and 
$\F\!\!\CONS \any_2$
rules are not analytic (i.e., the conclusion is not a subformula of
any premiss).
}
The proof is very similar to the \mci\/ \ke\/ system's soundness and completeness proof  presented in Section B.2.4 of \cite{teseAdolfo}. We begin with some definitions.

\begin{definition}\cite{DAgostino99}
A branch of a \ke~tableau is closed when $\T A$ and $\F A$ appear in the branch.
\end{definition}

\begin{definition}\cite{DAgostino99}
A \ke~tableau is closed if all its branches are closed.
\end{definition}

\begin{definition}
$\Gamma \vdash_{\mbox{\tiny \cum\ke}} B$ if there is a closed \ke~tableau for $\Gamma \vdash B$.
\end{definition}

\begin{definition}
The \cum\/ \ke\/ system  is sound if, for any $\Gamma$ and $B$, 
$\Gamma \vdash_{\mbox{\tiny \cum\ke}} B$  implies  $\Gamma\vdash_{\cum}B$.
\end{definition}

\begin{definition}
The \cum\/ \ke\/ system  is complete if, for any $\Gamma$ and $B$, 
$\Gamma\vdash_{\cum}B$ implies $\Gamma \vdash_{\mbox{\tiny \cum\ke}} B$.
\end{definition}

\com{

\begin{theorem}
The \cum\/ \ke\/ system is sound and complete.
{\em
\begin{pf}
The proof (obtained by following the style of the soundness and completeness proof 
for the analytic tableau system \cite{Smullyan68})
is presented in Appendix~\ref{ke_cum_proof_sc}.
\com{
Due to lack of space we omit the soundness and completeness proof for the \cum\/ \ke\/ System.
The proof can be obtained by following the style of the soundness and completeness proof 
for the analytic tableau system \cite{Smullyan68}.
}
\end{pf}
}
\end{theorem}

}

\begin{definition}
\rm
A set of \cum\/ signed formulas \satt\/ is {\em downward saturated}:
% is an \mbc\/ downward saturated set (see Definition~\ref{dowand_saturated} where   
% the additional clause if valid
\begin{enumerate}
\item whenever a signed formula is in \satt, its conjugate  is {\em not} in \satt;
\item when all premises of any \cum\/ \ke\/ rule  (except PB)
%(except `\pbNm') 
are in \satt, its conclusions are 
also in \satt;
\item when the major premiss of a 2-premiss \cum\/ \ke\/ rule is in \satt, either its auxiliary premiss or its 
conjugate is in \satt. 
%
%And if $\T \neg X$ is in \satt\/, either $\T \cons X$ or $\F \cons X$ can be in \satt\/,
%but only if $\cons X$ is a subformula of some other formula in \satt.
%If $\cons X$ is {\em not}  a subformula of some other formula in \satt,
%neither $\T \cons X$ nor $\F \cons X$ are in \satt;
%
% % % if $\T \neg X$ is in \satt\/
% % % and $\cons X$ is a subformula of some other formula in \satt, 
% % % then either $\T \cons X$ or $\F \cons X$ is in \satt\/;
% so para mci:
% \item if a signed formula ${\cal S}$ $X$ is in \satt,  then 
% for any sign ${\cal S}$, for any formula $X$, for all subformulas $Y$ of $X$ 
% and for all $n \ge 0$, the signed formula $\T \CONS\NOTBEF^n\CONS Y$ is in \satt.
\end{enumerate}
\label{downward_saturated}
\end{definition}

A Hintikka's Lemma holds for \cum\/ downward saturated sets:
\begin{lemma} 
% \rm
(Hintikka's Lemma for \cum) Every \cum\/ downward saturated set is satisfiable.
{\em 
\begin{proof}
For any downward saturated set \satt, we can easily construct a \cum\/ 
valuation $v$ 
such that for every signed formula 
${\cal S} X$ in the set, $v({\cal S} X)=1$. How can we guarantee 
this is in fact a valuation? First, we know that there is no pair
$\T X$ and $\F X$ in \satt.
Second, all premised \cum\/ \ke\/ rules preserve valuations.
That is, if $v({\cal S} X_i)=1$ for every premiss ${\cal S} X_i$, 
then $v({\cal S} C_j)=1$ for all conclusions $C_j$.
And if $v({\cal S} X_1)=1$ and $v({\cal S} X_2)=0$,
where $X_1$ and $X_2$ are, respectively, major and minor premises
of a \cum\/ \ke\/ rule, then $v({\cal S'} X_2)=1$, where ${\cal S'} X_2$
is the conjugate of ${\cal S} X_2$.
\com{
For instance, suppose $\T A \AND B \in $ \satt, then $v(\T A \AND B)=1$.
In accordance with the definition of downward saturated sets, 
$\{ \T A, \T B \} \subseteq$ \satt. 
And by the definition of \cum\/ valuation, $v(\T A \AND B)=1$ implies $v(\T A)=v(\T B)=1$.
}
Therefore, \satt\/ is satisfiable.
% 
% That is, if $v({\cal S} X_i)=1$ for every premise ${\cal S} X_i$, 
% then $v({\cal S} C_j)=1$ for all conclusions $C_j$.
% And if $v({\cal S} X_1)=1$ and $v({\cal S} X_2)=0$,
% where $X_1$ and $X_2$ are, respectively, major and minor premises
% of an \mbc\/ \ke\/ rule, then $v({\cal S'} X_2)=1$, where ${\cal S'} X_2$
% is the conjugate of ${\cal S} X_2$.
% For instance, suppose $\T A \AND B \in $ \satt, then $v(\T A \AND B)=1$.
% In accord with the definition of downward saturated sets, 
% $\{ \T A, \T B \} \subseteq$ \satt. 
% And by the definition of valuation, $v(\T A \AND B)=1$ implies $v(\T A)=v(\T B)=1$.
\end{proof}
}
\end{lemma}

\begin{theorem}
% \rm
Let \satt' be a set of signed formulas. 
\satt' is satisfiable {\em if and only if} there
exists a downward saturated set \satt'' such that \satt' $\subseteq$ \satt''.
\end{theorem}
\begin{proof}
$(\Leftarrow)$
First, let us prove that if there exists a downward saturated set \satt'' such that \satt' $\subseteq$ \satt'', then \satt' is satisfiable.
This is obvious because from \satt'' we can obtain a valuation that satisfies all
formulas in \satt'', and \satt' $\subseteq$ \satt''.

$(\Rightarrow)$
Now, let us prove that if \satt' is satisfiable, there exists a downward saturated 
set \satt'' such that \satt' $\subseteq$ \satt''.

So, suppose that \satt' is satisfiable and that there is no
downward saturated set \satt'' such that \satt' $\subseteq$ \satt''.
Using items (ii) and (iii) of \eqref{downward_saturated}, we can obtain 
a family of sets of signed formulas \satt'$_i$ ($i \ge 1$) that include \satt'. 
If none of them is downward saturated, it is because for all $i$, $\{ \T X, \F X\} \in
$\satt'$_i$ 
for some $X$.
But all rules are valuation-preserving, so this can only happen if 
\satt\/ is unsatisfiable, which is a contradiction. 
\end{proof}

\begin{corollary}
% \rm
\satt' is an unsatisfiable set of formulas if and only if there is no downward
saturated set \satt'' such that \satt' $\subseteq$ \satt''.
\end{corollary}

% Theorem~\ref{mbc_ccaproof_theorem_one} and Corollary~\ref{mbc_ccaproof_corollary_one}
% also hold for \cum\/ downward saturated sets.

\begin{theorem}
% \rm
The \cum\/ \ke\/ system is sound and complete.
{\em
\begin{proof}
The \cum\/ \ke\/ proof search procedure for a set of signed formulas $S$ 
either provides one or more downward saturated sets that give a 
valuation satisfying $S$ or finishes with no downward saturated set.
The \cum\/ \ke\/ system is a refutation system.
The \cum\/ \ke\/ system is sound because 
if a \cum\/ \ke\/ tableau for a set of formulas $S$ closes, then there is no 
downward saturated set that includes it, so $S$ is unsatisfiable.
If the tableau is open and completed, then any of its open branches can be represented 
as a downward saturated set and be used to provide a valuation that satisfies $S$
(in other words, $S$ is satisfiable).
% By construction, downward saturated sets for open branches are analytic, i.e.~include
% only subformulas of $S$. Therefore, the \mbc\/ \ke\/ system is analytic.

The \cum\/ \ke\/ system  is complete because if $S$ is satisfiable, no \cum\/ \ke\/ 
tableau for a set of formulas $S$ closes. And if $S$ is unsatisfiable, 
all completed \cum\/ \ke\/ tableaux for $S$ close.
\end{proof}
}
\end{theorem}

% \subsubsection{Important features of the \ke\/ System for \cum}
% 
% It is correct (sound) and complete.
% 
% This is easy to prove.
% Every rule in the CCCM's tableau system for \cum\/ is in or can be derived from our rules.
% And their system is correct and complete.
% 
% Another more direct proof is possible.
% It is in the appendix.
% 
% Our system, however, is not analytic.
% We cannot say that every \ke\/ System for \cum\/ is not analytic.

\subsection{Decidability}

We do not prove here that the \cum\/ \ke\/ system is decidable, i.e., that 
there is an algorithm for finding proofs in the \cum\/ \ke\/ system.
% any 
% proof search eventually terminates.
We  only  present the sketch of such a proof that will be detailed 
in a future paper about the implementation of a \cum\/ prover.

The idea is to define a restriction of the \cum\/ \ke\/ system which 
imposes some conditions on the application the PB rule (Figure~\ref{cum_KE_rules}).
In this {\em restricted} \cum\/ \ke\/ system, the PB rule can only be applied in a branch:
\begin{itemize}
 \item when there is a non-atomic signed formula %${\cal S} F$
that can be the main premiss of a 2-premiss rule and that was not yet analysed (i.e. used 
as main premiss) in the branch; and
 \item when either $\T A$ or $\F A$ can be the minor premiss of a 2-premiss rule,
where $A$ is the PB formula (i.e.~the $A$ formula that appears as $\T A$ in the new left branch and $\F A$ in the new right branch after PB application).
\end{itemize}

For all the 2-premiss rules in Figure~\ref{cum_KE_rules}, the minor premiss's size is smaller 
than major premiss's size. This, alongside with the conditions above, guarantees the the proof search procedure eventually terminates.

\section{A KEMS Strategy for \cum}
\label{strategy_cum}

KEMS \cite{teseAdolfo} is a theorem prover that can be used to implement strategies for many different logical systems. For instance, the current version \cite{KEMSsite} has 6 strategies for \cpl, 2 strategies for \mbc\/
and 2 strategies for \mci.
%Actually, it is not easy to implement a strategy for a logical system in \kems.

We have to follow some steps to implement a strategy for a logical system in \kems. 
First, one has to know how KEMS implementation is structured (by reading \cite{teseAdolfo} and the source code available in \cite{KEMSsite}).
Second, one has to implement the classes that will represent the logical system (such as \cpl\/ or \cum).
Third, one has to implement the classes necessary to represent the rules of the \ke\/ system (such as \cum\/ \ke\/ system).
Only after these three steps, one can implement one or more strategies for a given \ke\/ system.

\ke\/ systems (as well as many logical proof methods) are usually presented by showing their rules.
The rules tells us only what we can do -- they do not specify in which order to use the rules.
%how to perform a proof.
A strategy is a deterministic algorithm for a given \ke\/ system, as well
as a set of data structures used by the algorithm.

\com{
For \cum, a first action can be necessary for all strategies. If the input is a problem written in $\Sigma_\CONS$, one can rewrite it, by using \cumDef, so that all its formulas are in  $\Sigma$.
}

\subsection{\cum\/ \ke\/ Simple Strategy}
\label{cumSS}

%  This is an extension of \cpl\/ Simple Strategy for \mbc. 
% It implements the \embcke\/ system -- therefore it does not use any simplification
% rule.

The \cum\/ \ke\/ Simple Strategy resembles  \mbc\/  and \mci\/ Simple Strategies 
(see Sections C.4.4 and C.4.5 of \cite{teseAdolfo}). 
Let us informally describe the algorithm performed by this strategy:
% \begin{enumerate}
%  \item the strategy applies all possible linear rules in a branch until the branch is closed 
% (i.e.~when a contradiction is found) or linearly saturated (that is, when no more linear rules can be applied);
%  \item the strategy uses a stack of open branches. If a branch closes, the strategy removes a branch 
% from the stack and starts to apply rules on this branch (i.e.~goes back to step 1)
%  \item when a branch is found to be open and saturated, the procedure ends and the result is that the 
% tableau is open;
%  \item when there is no more branch in the stack, the procedure ends and the result is that the 
% tableau is closed.
% \end{enumerate}
\begin{enumerate}
 \item the strategy applies all possible linear rules in the current branch
(in the beginning, the current branch is the branch containing the formulas obtained from the problem);
 \item if the current branch closes (i.e.~if a contradiction $\{\T A, \F A\}$ is found), 
then the strategy tries to remove a branch from its {\em stack of open branches}. If it succeeds, 
this branch becomes the current branch and the control goes back to the first step. 
If there is no remaining open branch,  the procedure ends and the result is that the 
tableau is declared {\bf closed};
 \item if the current branch is linearly saturated (i.e.~no more linear rules can be applied),  
but not closed, the  strategy tries to apply the PB rule. 
The PB rule can be applied when there is at least one 
non-atomic signed formula in the branch %${\cal S} F$
that can be the main premiss of a 2-premiss rule and this signed formula was not yet 
used as the main premiss in an application of a 2-premiss rule.
If the strategy can apply the PB rule, then the (new) right
branch is put in the stack of open branches and the  left branch becomes the current branch.
If the strategy cannot apply the PB rule, then the procedure finishes by declaring the tableau {\bf open}.
\end{enumerate}

The order of rule applications 
%(rules from Figure~\ref{cum_KE_rules}) 
is:
\begin{enumerate}
  \item \cum~\ke~1-premiss rules;
  \item \cum~\ke~2-premiss rules;
  \item the  PB rule.%\pbNm\/
\end{enumerate}
See Sections C.2 and C.4 of \cite{teseAdolfo} for more details on how rules are applied in KEMS.

\subsubsection{Implementation Remarks}

This strategy is a very straightforward strategy for a \cum~\ke~system.
The idea is to use the PB rule only as a last resource (as shown in the canonical procedure for
\ke~\cite{DAgostino99}).
The difference is that in the \cum~\ke~system we cannot restrict the strategy to perform only analytic
applications of PB. An analytic application of PB is an application of PB where the PB formula
(i.e.~the $A$ formula that appears as $\T A$ in the new left branch and $\F A$ in the new right branch
after PB application)
is a subformula of some formula in the branch.

%I AM SURE THIS IS FALSE. I DO NOT KNOW WHEN NON-ANALYTIC PBs ARE NECESSARY.
% To prove sequents such as 
% \[
% A \AND B, \NOT (A \AND B) \vdash \NOT\NOT\NOT\CONS A \AND \NOT\NOT\NOT\CONS B
%  \]
% it is essential to perform non-analytic PBs.

Another difficulty in the implementation of this strategy (actually in the implementation of almost any proof
system for \cum) is how to deal with the consistency connective.

We have two options:
\begin{enumerate}
 \item only accept problems using the connectives in $\Sigma$. Therefore, all rules presented in
Figure~\ref{cum_KE_rules} will have to be implemented using $\Sigma$ connectives (which makes the
rules and the associated pattern matching more complex). Note that the size of 
problems written in $\Sigma^\CONS$ may grow exponentially (in the worst case) when translated to
$\Sigma$;
 \item accept problems written in $\Sigma^\CONS$ and, whenever a $\NOT (A \AND \NOT A)$ formula appears
(for any $A$), treat it as if it was (also) $\cons A$ in the applications of rules that have formulas with
$\CONS$ as premisses. Although this option allows the prover to deal with smaller problems, it makes rule applications more difficult. 
\end{enumerate}

\cum\/ Simple Strategy will use option (i) above. Option (ii) will be used on a second strategy for \cum~\ke.

%Why \ke~systems are more difficult to implement? 
\com{

It is easier to implement analytic tableau systems than \ke~systems
because in analytic tableau systems all rules have exactly one premiss.
There are no 0-premiss or 2-premiss rules.

In 2-premiss rules, there is a major and a minor premiss. When the strategy finds a formula that can be the major
premiss in the application of a 2-premiss rule, the strategy must try to find a formula in the current branch 
that can be the minor premiss for that rule.
If the strategy does not find a minor premiss, than the strategy 
than the strategy must look for another formula that can be a major premiss in a 2-premiss rule application.

Only when all possible 2-premiss rules in a branch are applied, the strategy is allowed to try to apply the 
0-premiss PB rule.

To apply the PB rule one usually needs a motivation. That is, all formulas that can be a major premiss
in a 2-premiss rule are called PB candidates.
To apply PB the strategy must choose a PB candidate (which is associated to a 2-premiss rule).
After that, the strategy branches on the current branch, producing two new branches.
The motivation for applying PB is to be able to apply the 2-premiss rule on the (new) left branch.
}

\com{

It is important to notice here that in \cum 's \tnegMbcNm\/ rule
the two premises have the same size.
% (in \cpl\/ two-premise rules 
%the major premise is always bigger than the minor premise).
%We have proved (see Section~\ref{tes_tp_KE_mbc_ccaproof})
%that 
We only need to branch on a `$\CONS A$' formula if it already appears 
as subformula of some formula in this branch. 
So we have to add this check before applying the \pbNm\/ rule.
%The other features are equal to Simple Strategy features.

Every PB application must have a motivation, that is, there must be 
a formula in the current branch which will be used as main premiss in some rule application R after PB.
This formula must have not yet been analyzed.
The formula in the left branch created by the PB application will be the minor premiss in the application 
of R.

%It implements the \ecumke\/ system.

???
But here we do not have to restrict the application of the \pbNm\/ rule
because of the \tnegMbcNm\/ rule.
All other features of this strategy are equal to \mbc\/ Simple Strategy features.

\subsection{Properties of \cum\/  \ke\/ Proofs}

\newcommand{\ThisTR}[1]{\Tr[ref=t]{#1}}

And what about the difference between \ke\/ \cpl\/ proofs and \ke\/ \cum\/ proofs?
As \cum\/ \ke\/ is an extension to \mbc\/ \ke, all rules that can be used with \mbc\/
can also be used with \cum.
But in \cum\/ \ke\/ we can also use  ($\T \neg \neg$), ($\F \cons 1$), and ($\F \cons 2$) rules.

Let us suppose again we are expanding ${\cal T}_1$ in Figure~\ref{proof_sketch}, now a \ke\/ \cum\/ proof tree, and we apply the PB rule on a `$\cons A$' formula. 
In addition to the possibilities described in the previous Section, 
in the left subtree ${\cal T}_2$, $\T \cons A$ can be used as a minor premiss
in applications of the  ($\F \cons 1$) and ($\F \cons 2$) rules,

And in the right subtree ${\cal T}_3$, 
because of \cumDef\/ we have the following:
\[
\begin{array}{l}
\F  \neg (A \AND \neg A) \\
\T (A \AND \neg A) \\
\T A \\
\T \neg A \\
\vdots
\end{array}
\]

In addition to that, if $A = (A_1 \oslash A_2)$, $\F \cons A$ can be used as a major premiss
in ($\F \cons 1$) and ($\F \cons 2$) applications.

% And if $A=\neg B$, then from $\T \cons A$, which using definition
% \cumDef\/ we can obtain:
% \[
% \begin{array}{l}
% \T \neg \neg (B \AND \neg B) \\
% \T (B \AND \neg B) \\
% \T B \\
% \T \neg B \\
% \F \cons B
% \end{array}
% \]

}

\section{Problem Families to Evaluate \cum\/ Provers}
\label{problemFamilies}

A problem family is a set of problems that we know, by construction, whether they are valid, satisfiable
or unsatisfiable \cite{teseAdolfo}. A problem is a sequent that can be given as input for a theorem prover. 
The $i$-th instance (for $i \ge 1$) of a problem family is a (valid, satisfiable or unsatisfiable) sequent.

In Section D.1.2 of \cite{teseAdolfo}, seven families of difficult problems that
can be used to evaluate theorem provers for paraconsistent logics were presented.
All these families were families of valid sequents.
To the best of our knowledge, there are no other families of difficult problems
designed with this purpose in mind.
The families presented there can be used to evaluate provers for two 
logics: \mbc\/ and \mci, which are part of 
the class of logics of formal inconsistency (LFIs) \cite{CCM05}.

In \cite{CCM05} it is shown that \cum\/ can also be classified as an LFI and that it extends \mbc.
Therefore, the first four families created to evaluate \mbc\/ provers \cite{teseAdolfo} can also be used to evaluate provers for \cum. 

However, these families do not test all \cum\/ \ke\/ rules. That is, to prove the problems in those families 
using the \cum\/ \ke\/ system, one does not need to use all its rules.
Therefore, to extend what we could call ``rule coverage'', i.e.~to test more rules, we present two more families of sequents.
These sequents are valid in \cum\/ (but not in \mbc, therefore they can also be used to test \mbc\/ provers) and, for proving them, we have to use rules that are not used in the first four families' proofs.

These families were not developed with any intuitive meaning in sight. As the objective was to test theorem
provers, they were designed to be difficult to prove, by using as many rules as possible.

The motivation for developing and presenting these problem families {\em before} the 
actual \cum\/ prover was implemented was, inspired by the Test-Driven Development technique for 
software development \cite{TDD-BECK-2002}, to use the tests as a guide for the design and 
implementation of the sofwtare.

Note: to make it easier to read the problems, we have used the connectives in $\Sigma^\CONS$ and we sometimes use ``['' and ``]'' in place of  ``('' and ``)''.

\subsection{Fifth family}

%\begin{quotation}
The sequents in this family (\FifF) demand \cum\/'s $\,\T\neg\neg\,$ rule to be proven valid.
%it is intended to distinguish between \mbc\/ and \cum\/ \ke\/ Systems. 
%\end{quotation}
$\Phi^5_n$ (the nth instance of \FifF) is:
\[
\CONS A_1,
\bigand^n_{i=1} (A_i), \bigand^n_{i=1} [ A_{n+1} \IMP  ( (A_i \OR B_i) \IMP (\cons A_{i+1}) ) ],
(\bigand^n_{i=1} \cons A_i) \IMP \neg A_{n+1} \\
\vdash \neg \neg \neg A_{n+1}
\]

\troca{2}
{}
{
For instance, $\Phi^5_3$ in signed tableau notation  is:
% \[
% \begin{array}{l}
% \T \, A_1 \AND A_2 \AND A_3 \\
% \T [A_4 \IMP ((A_1 \OR B_1) \IMP (\cons A_2))]  \AND [A_4 \IMP ((A_2 \OR B_2) \IMP (\cons A_3))] \AND  
%  [A_4 \IMP ((A_3 \OR B_3) \IMP (\cons A_4))] \\
% \T ( (\cons A_1)  \AND (\cons A_2) \AND (\cons A_3)) \IMP \neg A_4\\
% \F \neg \neg \neg A_4
% \end{array}
% \]
% 
\[
\begin{array}{l}
\T \,\, \cons A_1 \\
\T \,\, A_1 \AND A_2 \AND A_3 \\
\T \,\, [A_4 \IMP ((A_1 \OR B_1) \IMP (\cons A_2))]  \\
  	\,\,\,\, \AND \, [A_4 \IMP ((A_2 \OR B_2) \IMP (\cons A_3))]   \\
	\,\,\,\, \AND \, [A_4 \IMP ((A_3 \OR B_3) \IMP (\cons A_4))] \\
\T \,\, ( (\cons A_1)  \AND (\cons A_2) \AND (\cons A_3)) \IMP \neg A_4\\
\F \,\, \neg \neg \neg A_4 \\
\end{array}
\]
}

\subsection{Sixth family}

In order to prove, using the \cum\/ \ke\/ system, that the sequents in this family (\SF) are valid, 
it is necessary to use the two \cum\/ \ke\/ rules where ``$\CONS$'' is the main connective in the main premiss: $\F \cons \any_1$ and $\F \cons \any_2$.
% It is also essential to use the definition of ``$\cons$''.
% It is also essential to use the definition of $\cons$ (that is, if the the prover 
% accepts problems written in $\Sigma^\CONS$, it must be able to \cumDef).
$\Phi^6_n$ (the n-th instance of \SF) is:
\begin{small}
\[
\bigand^n_{i=1} (B_i), 
\bigand^n_{i=1} (\CONS C_i),
\bigand^n_{i=1} ( (A_i \OR B_i) \IMP (\cons A_{i+1}) ),
(\bigand^n_{i=1} C_i) \IMP (D \AND \neg C_1)
\vdash 
[\bigor^n_{i=1}(\cons(A_{i+1} \IMP C_i) )] \OR D
\]
\end{small}

\troca{2}
{}
{
For instance, $\Phi^6_3$ is:
\[
\begin{array}{l}
\T \, B_1 \AND B_2 \AND B_3 \\
\T \, \CONS C_1 \AND \CONS C_2 \AND \CONS C_3 \\
\T ((A_1 \OR B_1) \IMP (\cons A_2))  \AND ((A_2 \OR B_2) \IMP (\cons A_3)) \AND
  ((A_3 \OR B_3) \IMP (\cons A_4)) \\
\T (C_1 \AND C_2 \AND C_3) \IMP (D \AND \neg C_1) \\
\F [\cons(A_2 \IMP C_1)] \OR [\cons(A_3 \IMP C_2)] \OR [\cons(A_4 \IMP C_3)] \OR D
\end{array}
\]
}

\section{A Motivating Example}

%TODO: COMPARE MORE EXPLICITLY WITH CLASSICAL LOGIC.

We present here an example almost completely based on the example shown in \cite{KrauseBibel}:
\begin{quotation}
Consider the construction of a simple medical system aimed at diagnosing three diseases $K$, $L$ and $M$.
There are two different symptoms, denoted by $N$ and $O$. The intended usage of this system is as follows:
\begin{itemize}
\item The core part of the system is the knowledge provided by a doctor ($DOC_1$).
\item When we intend to apply this knowledge to a specific patient, other professionals conduct medical tests on this patient add the results of these tests to the knowledge base.
\item In order to use the system, we submit a goal to the program in a similar way as it is done in Prolog.
\end{itemize}

We assume that the system is written in the form of a finite set of formulas over \cum\/. Suppose that $DOC_1$ provided us the following five rules (formulas):
\begin{description}
 \item [$(F_1)$]  $K \IMP \NOT L$
 \item [$(F_2)$]  $L \IMP \NOT K$
 \item [$(F_3)$]  $K \IMP M$
 \item [$(F_4)$]  $N \IMP K$
 \item [$(F_5)$]  $O \IMP L$
\end{description}

Intuitively, the doctor is telling that:
\begin{itemize}
 \item An individual cannot have both diseases $K$ and $L$ ($F_1$ and $F_2$).
 \item If an individual has the disease $K$, them he has the disease $M$ ($F_3$) 
 \item If an individual has the symptom $N$, them he has the disease $K$ ($F_4$) 
 \item If an individual has the symptom $O$, them he has the disease $L$ ($F_5$) 
\end{itemize}
\end{quotation}

To exemplify the use of this knowledge base, we describe four situations. The first one is similar to a query to a Prolog program, while the other three explore the capacity of handling inconsistencies:

{\bf Case 1:} Suppose that the patient has symptom $N$ and we want to know if he has the disease  $K$ but not $L$.

To answer this query we must verify if 
\[
  F_1, F_2, F_3, F_4, F_5, N \vdash_{\cum} K \AND \NOT L
\]
is valid.
\troca{1}
{
As the \ke\/ proof for this sequent is a closed tableau, this sequent is valid.
It is also valid in classical logic.
}
{

The \ke\/ proof
%\footnote{More about the \ke\/ system for \cum\/ in Section~\ref{ke_cum}.} 
for this sequent (a closed tableau, which shows that the sequent is valid) is the following:
\begin{small}
\[
\begin{array}{c}
\T K \IMP \NOT L \\ 
\T L \IMP \NOT K \\
\T K \IMP M \\
\T  N \IMP K \\
\T O \IMP L \\
\T N \\
\F K \AND \NOT L \\ \hline
\T K \\
\F \NOT L \\
\T \NOT L \\
\x 
\end{array}
\]
\end{small}
}

{\bf Case 2:} Now suppose that the patient tested positive for symptoms $N$ and $O$, and we want to know if he has both diseases $K$ and $L$.

To answer this query we must verify if 
% \[
%   F_1, F_2, F_3, F_4, F_5, N, O \vdash_{C_1} K \AND L
% \]
 \begin{equation}
  F_1, F_2, F_3, F_4, F_5, N, O \vdash K \AND L
%  F_1, F_2, F_3, F_4, F_5, N, O \vdash_{C_1} K \AND L
 \label{Case2}
 \end{equation}
is valid.
\troca{1}
{
\eqref{Case2} is valid in \cum. In classical propositional logic, \eqref{Case2} is also valid. 
Actually, in classical propositional logic:
%$F_1, F_2, F_3, F_4, F_5, N, O \vdash_{\mbox{\small CPL}} B$ 
 \begin{equation}
F_1, F_2, F_3, F_4, F_5, N, O \vdash B
 \label{Case2b}
 \end{equation}
is valid for any formula $B$.
However, \eqref{Case2b} is not valid in \cum~ for any formula $B$ .
}
{
The \ke\/ proof for this sequent (a closed tableau) is the following:
\begin{small}
\[
\begin{array}{c}
\T K \IMP \NOT L \\ 
\T L \IMP \NOT K \\
\T K \IMP M \\
\T  N \IMP K \\
\T O \IMP L \\
\T N \\
\T O \\
\F K \AND L \\ \hline
\T K \\
\T L \\
\F L \\
\x 
\end{array}
\]
\end{small}

Notice that in classical propositional logic, $F_1, F_2, F_3, F_4, F_5, N, O \vdash_{\mbox{CPL}} B$ for any formula $B$.
}

{\bf Case 3:} Now suppose that the patient tested positive for symptoms $N$ and $O$, and we want to know if he has {\em not} the disease $M$.

To answer this query we must verify if 
\begin{equation}
  F_1, F_2, F_3, F_4, F_5, N, O \vdash_{\cum} \NOT M
\label{Case3}
\end{equation}
is valid.
\troca{2}
{
The \ke\/ proof for this sequent is an open tableau, which shows that \eqref{Case3} is NOT valid.
}
{
The \ke\/ proof for this sequent (an open tableau, which shows that the sequent is NOT valid) is the following:
\begin{small}
\[
\begin{array}{c}
\T K \IMP \NOT L \\ 
\T L \IMP \NOT K \\ %
\T K \IMP M \\ %
\T N \IMP K \\ %
\T O \IMP L \\ %
\T N \\
\T O \\
\F \NOT M \\ \hline
\T K \\
\T L \\
\T M  \\
\T \NOT K \\
\T \NOT L \\
\end{array}
\]
\end{small}

However, this sequent is valid in classical logic, because a classical contradiction is found ($\T K$ and $\T \NOT K$).
}

{\bf Case 4:} Now suppose again that the patient tested positive for symptoms $N$ and $O$, but now we want to know if he has the disease $K$ and if this conclusion is  not consistent ($\NOT \cons K$).

To answer this query we must verify if 
\begin{equation}
   F_1, F_2, F_3, F_4, F_5, N, O \vdash_{\cum} K \AND \NOT \cons K
  \label{Case4}
\end{equation}
% \[
%   F_1, F_2, F_3, F_4, F_5, N, O \vdash_{C_1} K \AND \NOT \cons K
% \]
is valid.
\troca{2}
{
There is a closed \ke~ tableau for \eqref{Case4}, showing that it is a valid sequent. 

This query shows that, besides ``dealing with inconsistencies in the knowledge base without every formula becoming derivable'' \cite{KrauseBibel}, a common feature of paraconsistent logics, \cum\/ allows us to express propositions about the (in)consistency of formulas.
}
{
The \ke\/ proof for this sequent (a closed tableau) is the following:
\begin{small}
\[
\begin{array}{c}
\T K \IMP \NOT L \\ 
\T L \IMP \NOT K \\ %
\T K \IMP M \\ %
\T N \IMP K \\ %
\T O \IMP L \\ %
\T N \\
\T O \\
\F K \AND \NOT \cons K \\ \hline
\T K \\
\T L \\
\T \NOT L  \\
\T \NOT K \\
\T M \\
\F \NOT \cons K \\
\T \cons K \\
\F K \\
\x
\end{array}
\]
\end{small}

This query shows that, besides ``dealing with inconsistencies in the knowledge base without every formula becoming derivable'' \cite{KrauseBibel}, a common feature of paraconsistent logics, \cum\/ allows us to express propositions about the (in)consistency of formulas.

%This sequent is not representable in classical logic, because ``$\CONS$'' is not part of \cpl's language.

The sequent \eqref{Case4} is valid in classical logic. Note that, for any formula $B$, ``$\NOT \cons B$'' is a theorem in classical logic. Therefore, in classical logic, $\Gamma \vdash K \AND \NOT \cons K$ if and only 
$\Gamma \vdash K$.

}

%Why is \cum\/ important for something?

%\cite{DBLP:conf/issads/HasegawaAS05}

%A full proof based on Krause's examples? 

% Let us present a fact that is going to be important when we discuss tableau systems for \cum.
% The original set of clauses characterizing
% \cum-valuations (cf. \cite{CCM04}) had, instead of {\bf (v5)}:
% \begin{description}
% 	\item[(vC6)] $v(\cons \mybeta)=v(\myalpha \IMP \mybeta)=v(\myalpha \IMP \neg \mybeta)=1$
% implies $v(\myalpha)=0$
% \end{description}
% 
% However, having {\bf (v3)}, {\bf (v5)} in Definition~\ref{mbc_valuation}
% is equivalent to {\bf (vC6)}. The proof is as follows.
% First, let us show that if {\bf (v5)} and {\bf (v3)} then {\bf (vC6)} holds.
% Suppose {\bf (vC6)} does not hold, then we have 
% $v(\CONS \mybeta)=v(\myalpha \IMP \mybeta)=v(\myalpha \IMP \NOT \mybeta)=1$
% and $v(\myalpha)=1$ for some $\myalpha$ and $\mybeta$.
% As $v(\CONS \mybeta)=1$ then either $v(\mybeta)=0$ or $v(\NOT \mybeta)=0$.
% If $v(\mybeta)=0$, then $v(\myalpha \IMP \mybeta)=0$ since $v(\myalpha)=1$.
% If $v(\NOT \mybeta)=0$, then $v(\myalpha \IMP \NOT \mybeta)=0$ since $v(\myalpha)=1$.
% Therefore, we have a contradiction and {\bf (vC6)} must hold.
% 
% Now, let us show that if {\bf (vC6)} and {\bf (v3)} then {\bf (v5)} holds.
% Suppose {\bf (v5)} does not hold, then we can have 
% $v(\CONS \mybeta)=v(\mybeta)=v(\NOT \mybeta)=1$ for some $\mybeta$.
% But then, for any $\myalpha$,  $v(\myalpha \IMP \mybeta)=v(\myalpha \IMP \NOT \mybeta)=1$ (by {\bf (v3)})
% that leads (by {\bf (vC6)}) to $v(\myalpha)=0$. If $\myalpha$ is $\mybeta$ then $v(\mybeta)=0$, a contradiction.

\section{Related Work}

%Here we discuss the \cum\/ tableau systems described in the literature.

A tableau system for \cum\/ was presented in \cite{DBLP:journals/jancl/CarnielliM92}. 
As this system is based on analytic tableaux (AT) \cite{Smullyan68}, it has four branching rules:
the three ones from AT plus a $\T \, \NOT$ branching rule.
Due to this $\T \, \NOT$ rule, infinite loops may occur during the proof search, 
postponing indefinitely the analysis of formulas that involve the negation and consistency operators. Notwithstanding, this system is decidable. This system has been implemented but the source code is not available. 
%Due to its resemblance to AT, the \cum\/ \ke\/ system is probably more efficient than this system.

In \cite{Buchsbaum1993} two tableau systems for \cum\/ were presented, the second one being a version
of the first one considered by the authors more adequate to be implemented.
The first system has 12 rules (8 of them are branching rules) while the second has 20 rules (12 of them are branching rules). The rules are rather complex, involving much more formulas and connectives than \cum\/ \ke\/ rules. The second system was elegantly implemented in LISP (the source code is available in \cite{BuchsbaumMSc}). However it was written in a LISP dialect (muLISP) which cannot be compiled in
modern LISP compilers.

We have experimented using Buchsbaum's system with the formulae described in Section~\ref{problemFamilies}. 
For example, it was not able to prove instance $\Phi^5_{27}$ due to lack of memory. 
This confirms that the family $\Phi^5$ is a family of difficult problems.
We are translating this prover to a different LISP dialect to make it more robust.
% We intend to compare the results obtained by this prover with the results obtained
% by the KEMS strategies we are going to implement.

Another \cum\/ tableau system appears in \cite{CCM05}. 
%This system does
%not present loops. However, its proofs are less concise compared with those obtained with the 
%previous system.
%In \cite{CCM05}, the authors presented 
It was 
%is a sound and complete tableau system for \cum\/ 
obtained by using a general method for constructing tableau systems \cite{Caleiro2005}. 
Although this system has a PB branching rule, a feature of \ke\/ systems, 
is not a \ke\/ system.
To be a \ke\/ system it should have only one branching rule, but it  has 8 branching rules. 
Just like the system in \cite{DBLP:journals/jancl/CarnielliM92}, it is based on AT.
However, it does not have rules that lead to infinite loops.
We do not know of any implementation of this method.

%As explained in \cite{DAgostino99}, branching rules lead to inefficiency in the proof search procedure.
%Therefore, to obtain a more efficient proof system, we used that tableau system as a basis to devise an 
%original \cum\/ \ke\/ system\footnote{We had already presented \ke\/ systems for two other paraconsistent logics: \mbc\/ anc \mci\/ \cite{teseAdolfo}.}. 

In \cite{D'Ottaviano200627}, tableau systems for several logics of the $C_n$ hierarchy were presented.
The \cum\/ tableau system presented there is also based on AT.
While in the previous systems $\circ A \stackrel{\mbox{\tiny def}}{=} \neg (A \AND \neg A)$
was applied whenever necessary to generate the branches of the tableau, this system has specific rules
to directly deal with all operators, including ``$\CONS$''.
However, as it is based on the analytic tableau method, it also has too many (six) branching rules.
We also do not know of any implementation of this method.

\com{
It is important to notice that tableau systems presented in \cite{}
\cite{DBLP:journals/jancl/CarnielliM92}

developed by logicians

whilst

were developed by computer scientists

}

Therefore, the distinctive feature of our \cum\/ \ke\/ system is that
it has 13 essential rules and only of them is a branching rule. 
This feature will allow us to implement efficient strategies for this system 
in KEMS \cite{KEMSsite}.

%be able to implemented

\com{
Compared to \mbc\/ and \mci, \cum\/ has some features that make it more difficult to develop  a \ke\/ system
(and implement it). First, in \cum, differently from what happens in \mbc\/ and \mci, $\cons$  is a defined connective.
%`$\cons$'\/ rules are not easily obtained from $\NOT$ and $\AND$ rules.

There are many tableaux systems for \cum.

In \cite{DBLP:journals/jancl/CarnielliM92}, we have a tableau system with a looping rule.

In \cite{CCM05}, another tableau system (obtained via a uniform procedure) was presented.
It is based on analytic tableaux and has too many branching rules.

In \cite{D'Ottaviano200627}, tableau systems for several logics of the $C_n$ hieararchy are presented.
Some rules.
}

\section{Conclusion}

%\subsection{Further Work}

%Implement at least two strategies for \cum\/ in \kems\/ using the \cum\/ \ke\/ system.

%After 

%Compare the results obtained with 

%Arthur Buchsbaum's \cum\/ prover \cite{BuchsbaumMSc}

% COMPARE TO ARTHUR'S SYSTEM USING PROBLEM FAMILIES

% MATRIX CONNECTION METHOD FOR C1??

In this paper, we have presented a sound and complete \ke~system for Da Costa's \cum~calculus 
for paraconsistent logic.
We have shown that our system has less branching rules than other tableau systems for 
\cum~described in the literature
\cite{BuchsbaumMSc,Buchsbaum1993,CCM05,DBLP:journals/jancl/CarnielliM92,D'Ottaviano200627}.
Therefore, it is probably more efficient than those systems (see \cite{DAgostino99} for
a discussion on why branching leads to inefficiency).

We have also described a strategy for this \ke~system that can be implemented in \kems.
Future work includes implementing this strategy, as well as designing and implementing other
strategies for the \cum~\ke~system.

In order to evaluate \cum~\ke~strategies, we have developed two problem families. These families
and the first four problem families described in section D.1.2 of \cite{teseAdolfo}
can also be used to evaluate other theorem provers for \cum, such as Arthur Buchsbaum's prover
for \cum~\cite{BuchsbaumMSc}.

As further work, we intend to compare the results obtained by our strategies (in the 
style of section D.2 of \cite{teseAdolfo}) among themselves as well as with Arthur Buchsbaum's prover.


\begin{thebibliography}{10}

\bibitem{TDD-BECK-2002}
Kent Beck.
\newblock {\em Test Driven Development: By Example}.
\newblock {Addison-Wesley Professional}, November 2002.

\bibitem{broda95solution}
Krysia Broda, Marcello D'Agostino, and Marco Mondadori.
\newblock {A Solution to a Problem of Popper}.
\newblock In {\em {Proceedings of the conference Karl Popper Philosopher of
  Science}}, 1995.
\newblock
  \url{http://citeseerx.ist.psu.edu/viewdoc/summary?doi=10.1.1.43.7542}. Last
  accessed, June 2009.

\bibitem{BuchsbaumMSc}
Arthur Buchsbaum.
\newblock An automatic proof method for paraconsistent logic (in portuguese),
  1988.
\newblock Available at \url{http://migre.me/gQD}. Last accessed, Mar 2009.

\bibitem{Buchsbaum1993}
Arthur Buchsbaum and Tarcisio Pequeno.
\newblock A reasoning method for a paraconsistent logic.
\newblock {\em Studia Logica}, 52(2):281--289, June 1993.

\bibitem{Caleiro2005}
Carlos Caleiro, Walter Carnielli, Marcelo Coniglio, and Joao Marcos.
\newblock Two's company: The humbug of many logical values.
\newblock In {\em Logica Universalis}, pages 169--189. Birkhauser Basel, 2005.

\bibitem{CCM05}
Walter Carnielli, Marcelo~E. Coniglio, and Joao Marcos.
\newblock {\em {Handbook of the Philosophical Logic}}, volume~14, chapter
  {Logics of Formal Inconsistency}, pages 15--107.
\newblock {Springer-Verlag}, second edition, 2007.

\bibitem{DBLP:journals/jancl/CarnielliM92}
Walter~Alexandre Carnielli and Mamede Lima-Marques.
\newblock Reasoning under inconsistent knowledge.
\newblock {\em Journal of Applied Non-Classical Logics}, 2(1), 1992.

\bibitem{teseDaCosta}
Newton C.~A. da~Costa.
\newblock {\em {Sistemas Formais Inconsistentes}}.
\newblock Rio de Janeiro, NEPE, 1963.
\newblock Reprinted by Editora da UFPR, Curitiba, 1993.

\bibitem{daco:asem77}
Newton C.~A. da~Costa and E.~H. Alves.
\newblock A semantical analysis of the calculi {C$_n$}.
\newblock {\em Notre Dame Journal of Formal Logic}, 18(4):621--630, 1977.
\newblock Available at \url{http://migre.me/gMA}. Last accessed, Mar 2009.

\bibitem{CKB}
Newton C.~A. da~Costa, Decio Krause, and Otavio Bueno.
\newblock {\em {Handbook of the Philosophy of Science. Philosophy of Logic}},
  chapter {Paraconsistent Logics and Paraconsistency}, pages 791--911.
\newblock {Elsevier}, 2007.

\bibitem{dagostino-are}
Marcello D'Agostino.
\newblock {Are Tableaux an Improvement on Truth-Tables? Cut-Free proofs and
  Bivalence}.
\newblock {\em Journal of Logic, Language and Information}, pages 235--252,
  1992.
\newblock Available at \url{http://citeseer.nj.nec.com/140346.html}. Last
  accessed, May 2005.

\bibitem{DAgostino99}
Marcello D'Agostino.
\newblock Tableau methods for classical propositional logic.
\newblock In Marcello~D'Agostino et~al., editor, {\em Handbook of Tableau
  Methods}, chapter~1, pages 45--123. Kluwer Academic Press, 1999.

\bibitem{dagostino94taming}
Marcello D'Agostino and Marco Mondadori.
\newblock The taming of the cut: Classical refutations with analytic cut.
\newblock {\em Journal of Logic and Computation}, pages 285--319, 1994.

\bibitem{Carvalho}
Fabio~Romeu de~Carvalho, Israel Brunstein, and Jair~Minoro Abe.
\newblock {Prevision of Medical Diagnosis Based on Paraconsistent Annotated
  Logic}.
\newblock {\em International Journal of Computing Anticipatory Systems},
  18:288--297, 2005.

\bibitem{D'Ottaviano200627}
Itala M.~Loffredo D'Ottaviano and Milton~Augustinis de~Castro.
\newblock Analytical tableaux for da costa's hierarchy of paraconsistent
  logics.
\newblock {\em Electronic Notes in Theoretical Computer Science}, 143:27 -- 44,
  2006.
\newblock Proceedings of the 12th Workshop on Logic, Language, Information and
  Computation (WoLLIC 2005).

\bibitem{10.1109/CBMS.2007.71}
Fahim~T. Imam, Wendy MacCaull, and Margaret~Ann Kennedy.
\newblock Merging healthcare ontologies: Inconsistency tolerance and
  implementation issues.
\newblock {\em Proceedings of the Twentieth IEEE International Symposium on
  Computer-Based Medical Systems}, pages 530--535, 2007.

\bibitem{KrauseBibel}
Decio Krause, Emerson~Faria Nobre, and Martin Musicante.
\newblock Bibel's matrix connection method in paraconsistent logic: general
  concepts and implementation.
\newblock In {\em Proceedings of the XXI International Conference of the
  Chilean Computer Science Society}, pages 161--167, 2001.

\bibitem{teseAdolfo}
Adolfo Neto.
\newblock {\em A Multi-Strategy Tableau Prover}.
\newblock PhD thesis, University of Sao Paulo, 2007.
\newblock Available at \url{http://www.dainf.ct.utfpr.edu.br/~adolfo/Thesis/}.
  Last accessed, Mar 2009.

\bibitem{IFIPAI2006}
Adolfo Neto and Marcelo Finger.
\newblock {Effective Prover for Minimal Inconsistency Logic}.
\newblock In {\em {Artificial Intelligence in Theory and Practice}}, {IFIP},
  pages {465--474}. {Springer Verlag}, 2006.
\newblock Available at
  \url{http://www.springerlink.com/content/b80728w7m6885765}. Last accessed,
  November 2006.

\bibitem{KEMSsite}
Adolfo Neto and Marcelo Finger.
\newblock {\em {KEMS - A \ke-based Multi-Strategy Tableau Prover}}, 2006.
\newblock \url{http://www.dainf.ct.utfpr.edu.br/~adolfo/KEMS}. Last accessed,
  April 2009.

\bibitem{NF07}
Adolfo {Neto} and Marcelo Finger.
\newblock A {KE} tableau for a logic of formal inconsistency.
\newblock In {\em Proceedings of TABLEAUX'07 position papers and Workshop on
  Agents, Logic and Theorem Proving. Technical Report (LSIS.RR.2007.002) of the
  LSIS/Université Paul Cézanne}, Marseille, France, 2007.

\bibitem{Smullyan68}
Raymond~M. Smullyan.
\newblock {\em First-Order Logic}.
\newblock Springer-Verlag, 1968.

\bibitem{4153157}
Cláudio~Rodrigo Torres, Germano Lambert-Torres, Luiz Eduardo~Borges da~Silva,
  and Jair~Minoro Abe.
\newblock Intelligent system of paraconsistent logic to control autonomous
  moving robots.
\newblock In {\em IEEE Industrial Electronics, IECON 2006 - 32nd Annual
  Conference on}, pages 4009--4013, Nov. 2006.

\end{thebibliography}
\end{document}